\let\accentvec\vec
\documentclass[runningheads,a4paper]{llncs}

\let\vec\accentvec
\usepackage{amsmath}								
\usepackage{amssymb}
\usepackage[misc,geometry]{ifsym}    
\usepackage{makeidx}  
\usepackage{lipsum}
\usepackage{lastpage}
\usepackage{clrscode}
\usepackage{bm} 
\usepackage{multirow}
\usepackage{array}
\usepackage{graphicx}
\usepackage{picinpar}
\usepackage{enumerate}
\usepackage{epsfig}
\usepackage{boxedminipage}
\usepackage{indentfirst}
\usepackage{afterpage}
\usepackage{yhmath}
\usepackage{graphicx}	
\usepackage{diagbox}
\usepackage{extarrows}
\usepackage{color}
\usepackage{float}
\usepackage[colorlinks, linkcolor=blue, anchorcolor=blue, citecolor=blue]{hyperref}
\usepackage{setspace}
\usepackage[dvipsnames]{xcolor}
\usepackage{algorithm,algorithmic}
\usepackage{booktabs}
\usepackage{longtable}
\usepackage{tabu}

\begin{document}
\numberwithin{equation}{section}

\title{FDR-HS: An Empirical Bayesian Identification of Heterogenous Features in Neuroimage Analysis}
\titlerunning{FDR-HS: an Empirical Bayesian identification of heterogenous features}  

\author{Xinwei Sun\inst{1,6} \and Lingjing Hu\inst{2}(\Letter) \and Fandong Zhang\inst{3,6}  \and Yuan Yao\inst{4}(\Letter) 
\and Yizhou Wang\inst{5,6}} 
\institute{School of Mathematical Science, Peking University, Beijing, 100871, China \and Yanjing Medical College, Capital Medical University, Beijing, 101300, China \and Key Laboratory of Machine Perception (Ministry of Education), Department of Machine Intelligence, School of Electronics Engineering and Computer Science,Peking University, Beijing 100871, China \and Hong Kong University of Science and Technology and Peking University, China \and National Engineering Laboratory for Video Technology, Key Laboratory of Machine Perception, School of EECS, Peking University, Beijing, 100871, China \and Deepwise Inc., Beijing, 100085, China }
\authorrunning{Xinwei Sun et al.} 

\maketitle     
\begin{abstract}
Recent studies found that in voxel-based neuroimage analysis, detecting and differentiating ``procedural bias" that are introduced during the preprocessing steps from lesion features, not only can help boost accuracy but also can improve interpretability. To the best of our knowledge, GSplit LBI is the first model proposed in the literature to simultaneously capture both procedural bias and lesion features. Despite the fact that it can improve prediction power by leveraging the procedural bias, it may select spurious features due to the multicollinearity in high dimensional space. Moreover, it does not take into account the heterogeneity of these two types of features. In fact, the procedural bias and lesion features differ in terms of volumetric change and spatial correlation pattern. To address these issues, we propose a ``two-groups" Empirical-Bayes method called ``FDR-HS" (False-Discovery-Rate Heterogenous Smoothing). Such method is able to not only avoid multicollinearity, but also exploit the heterogenous spatial patterns of features. In addition, it enjoys the simplicity in implementation by introducing hidden variables, which turns the problem into a convex optimization scheme and can be solved efficiently by the expectation-maximum (EM) algorithm. Empirical experiments have been evaluated on the Alzheimer's Disease Neuroimage Initiative\thinspace(ADNI) database. The advantage of the proposed model is verified by improved interpretability and prediction power using selected features by FDR-HS. 
\end{abstract}

Dedicated to Professor Bradley Efron on the occasion of his 80th birthday. 

\section{Introduction}

In recent years, the issue of model interpretability attracts an increasing attention in voxel-based neuroimage analysis of disease prediction, e.g. \cite{haufe2014interpretation,biessmann2012interpretability}. Examples include, but not limited to, the preprocessed features on structural Magnetic Resonance Imaging (sMRI) images that usually contain the following voxel-wise features: (1) lesion features that are contributed to the disease (2) procedural bias introduced during the preprocessing steps and shown to be helpful in classification \cite{sun2017gsplit,vbm2001} (3) irrelevant or null features which are uncorrelated with disease label. Our goal is to stably select non-null features, i.e. lesion features and procedural bias with high power/recall and low false discovery rate (FDR).

The lesion features have been the main focus in disease prediction. In dementia disease such as Alzheimer's Disease (AD), such features are thought to be geometrically clustered in atrophied regions (hippocampus and medial temporal lobe etc.), as shown by the red voxels in Fig.~\ref{figure:illustrationAB} (A). To explore such spatial patterns, multivariate models with Total Variation \cite{TV} regularization can be applied by enforcing smoothness on the voxels in neighbor, e.g. the $n^{2}$GFL \cite{n2gfl} can stably identify the early damaged regions in AD by harnessing the lesions.

Recently, another type of features called procedural bias, which are introduced during the preprocessing steps, are found to be helpful for disease prediction \cite{sun2017gsplit}. Again, taking AD as an example, the procedural bias refer to the mistakenly enlarged Gray Matter (GM) voxels surrounding locations with cerebral spinal fluid (CSF) spaces enlarged, e.g. lateral ventricle, as shown in Fig.~\ref{figure:illustrationAB} (A). This type of features has been ignored in the literature until recently, when the GSplit LBI \cite{sun2017gsplit} was targeted on capturing both types of features via a split of tasks of TV regularization (for lesions) and disease prediction with general linear model (with procedural bias). By leveraging such bias, it can outperform models which only focus on lesions in terms of prediction power and interpretability. 

However, GSplit LBI may suffer from inaccurate feature selection due to the following limitations in high dimensional feature space: \footnote{Please refer supplementary material for detailed and theoretical discussion}: (1) multicollinearity: high correlation among features in multivariate models \cite{tu2005problems}; (2) ``heterogenous features": the procedural bias and lesion features differ in terms of volumetric change (enlarged v.s. atrophied) and particularly spatial pattern (surroundingly distributed v.s. spatially cohesive). Specifically, the multicollinearity could select spurious null features which are inter-correlated with non-nulls. Moreover, GSplit LBI fails to take into account the heterogeneity since it enforces correlation on features without differentiation. Such problems altogether may result in inaccurate selection of non-nulls, especially procedural bias. As shown in Fig.~\ref{figure:illustrationAB} (B) and Table~\ref{table:mdc}, the procedural bias selected by GSplit LBI are unstably scattered on regions that are less informative than ventricle. Moreover, the collinearity among features tends to select a subset of features among correlated ones, as discussed in \cite{elasticnet}. Such a limitation leads to the ignorance of many meaningful regions (such as medial temporal lobe, thalamus etc.) of GSplit LBI in selecting lesion features, as identified by the purple frames of FDR-HS in Fig.~\ref{figure:illustrationAB} (B). Moreover, the two problems above may get worse as dimensionality grows. In our experiments with a fine resolution ($4\times4\times4$ of 20,091 features), the prediction accuracy of GSplit LBI deteriorates to $89.77\%$ (as shown in Table~\ref{sec:experimental}), lower than $90.91\%$ reported in \cite{sun2017gsplit} with a coarse resolution ($8\times8\times8$ of 2,527 features). 
\begin{figure}
\centering
\begin{minipage}[t]{0.64\linewidth}
    \includegraphics[width= \columnwidth]{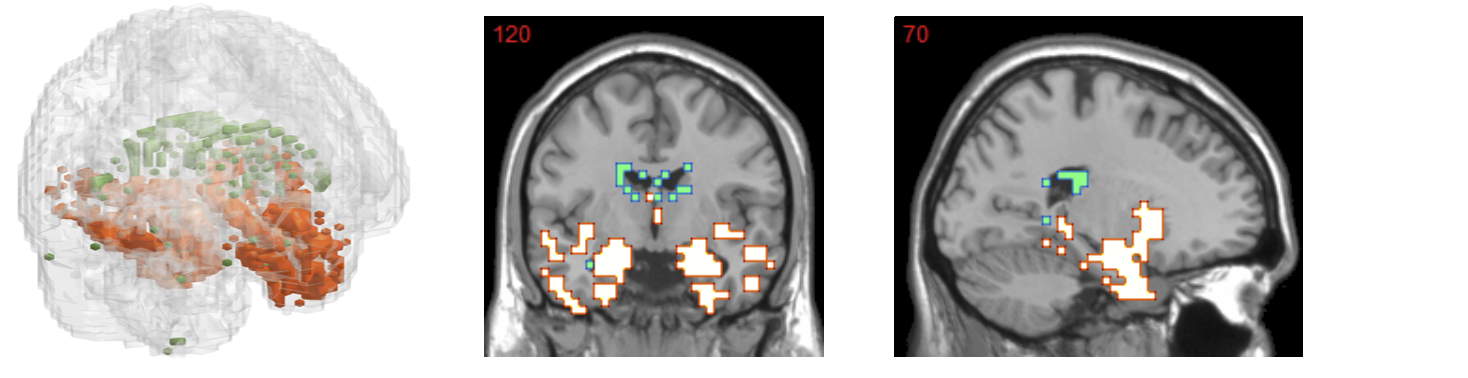}\begin{center}{A}\end{center}
    \end{minipage}
\begin{minipage}[t]{0.291\linewidth}
    \includegraphics[width= \columnwidth]{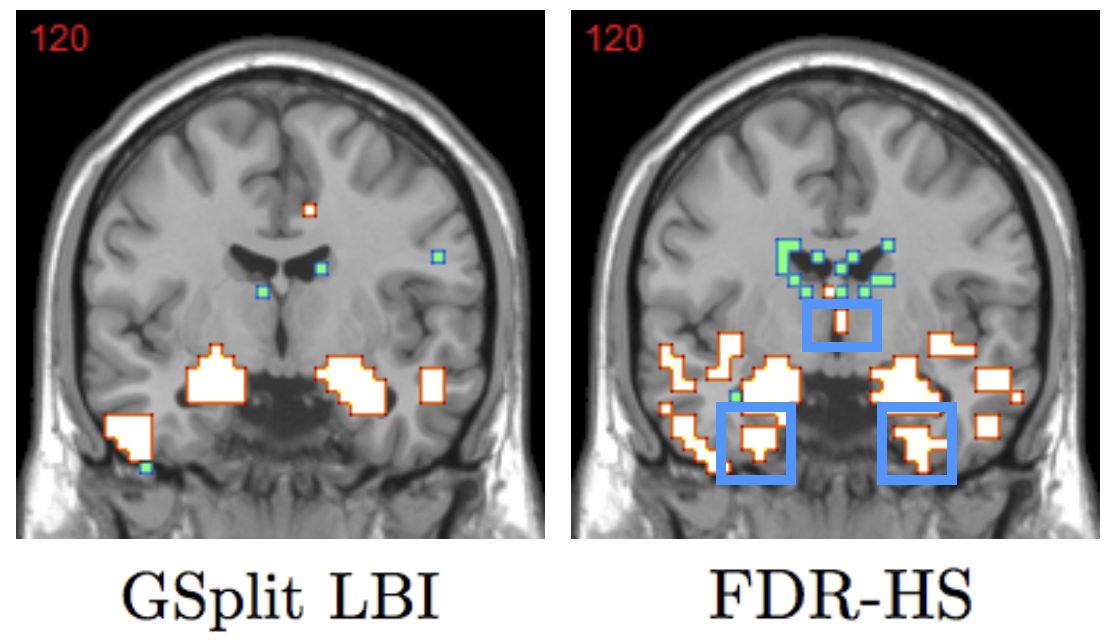}\begin{center}{B}\end{center}
    \end{minipage}
\caption{A: the features selected by FDR-HS (green denotes procedural bias; red denotes lesion features which are geometrically clustered) B: comparison with GSplit LBI }\label{figure:illustrationAB}
\end{figure}

To resolve the problems above, we propose a ``two-groups" empirical Bayes method to identify heterogenous features, called FDR-HS standing for ``FDR Heterogenous smoothing" in this paper. As a univariate FDR control method, it avoids the collinearity problem by proceeding voxel-by-voxel, as discussed in \cite{efron2016computer}. Moreover, it can deal with heterogeneity by regularizing on features with different levels of spatial coherence in different feature groups, which remedies the problem of losing spatial patterns that most conventional mass-univariate models suffer from, such as two sample T-test, BH$_{q}$ \cite{benjamini1995controlling} and LocalFDR \cite{efron2016computer}. 
By introducing a binary latent variable, our problem turns into a convex optimization and can be solved efficiently via EM algorithm like \cite{tansey2017false}. 
The method is applied to a voxel-based sMRI analysis for AD with a fine resolution (4$\times$4$\times$4 of 20,091 features). As a result, our proposed method exhibits a much stabler 
feature extraction than GSplit LBI, and achieves much better classification accuracy at $91.48\%$. 

 \section{Method}
Our dataset consists of $p$ voxels and $N$ samples $\{x_{i},y_{i}\}_{1}^{N}$ where $x_{ij}$ denotes the intensity value of the $j^{th}$ voxel of the $i^{th}$ sample and $y_{i} = \{\pm 1\}$ indicates the disease status ($-1$ denotes AD). The FDR-HS method is proposed to select non-null features. Such method is the combination of ``two-groups" model and heterogenous regularization, which is illustrated in Fig.~\ref{figure:illustration} and discussed below. 
\begin{figure}[!h]
  \centering
  \includegraphics[width = 0.9651\textwidth]{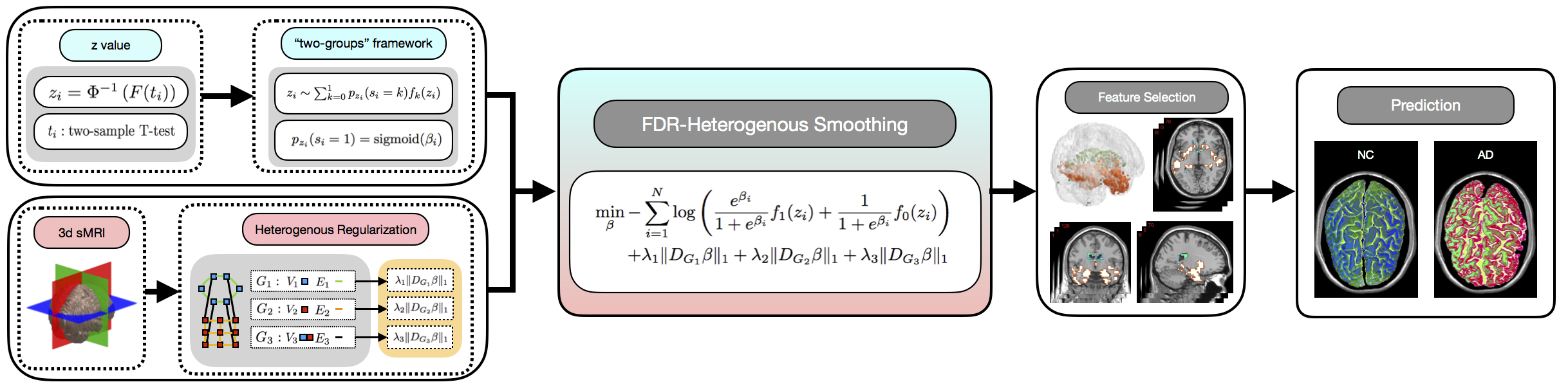}
   \caption{Illustration of FDR-HS model.}
 \label{figure:illustration}
 \end{figure}
\noindent
\textbf{Model Formulation.} Assuming for each voxel $i \in \{1,...,p\}$, the statistic $z_{i}$ is sampled from the following mixture:
\begin{align}
\label{eq-z}
z_{i} & \sim \sum_{k=0}^{1} \mathrm{p}(s_{i} = k)\mathrm{p}(z_{i} | s_{i} = k) = c_{i} f_{1}(z_{i}) + (1 - c_{i})f_{0}(z_{i}), 
\end{align} 
where $s_{i}$ is a latent variable indicating if the voxel $i$ belongs to the group of null features ($s_i=0$) or the group of non-null ones ($s_i=1$), $c_i = \mathrm{p}(s_{i} = 1) = \mathrm{sigmoid}(\beta_{i}) = e^{\beta_{i}} / \left( 1 + e^{\beta_{i}}\right)$ and $z_{i} =  \mathrm{\Phi}^{-1}\left(F_{N-2}(t_{i})\right)$ with $t_{i}$ computed by two-sample $t$-test.
Correspondingly, $f_{0}(\cdot)$ is density function of nulls, i.e. uncorrelated with AD and $f_{1}(\cdot)$ is that of non-nulls, i.e. procedural bias and lesions. The loss function can thus be defined as negative log-likelihood of $z_{i}$:
\begin{align}
\ell(\beta) = -\sum_{i=1}^{N} \log{\left(\frac{e^{\beta_{i}}}{1 + e^{\beta_{i}}}f_{1}(z_{i}) + \frac{1}{1 + e^{\beta_{i}}}f_{0}(z_{i})\right)}
\label{eq:loss-first}
\end{align}
which can be viewed as logistic regression (when $f_{0}$ and $f_{1}$ are replaced with binaries, as~\eqref{eq:loss-h}) with identity design matrix since~\eqref{eq-z} proceeds voxel-by-voxel. Hence, it does not have the problem of multicollinearity.

\noindent
\textbf{Selecting Features.} To select features, we compute the posterior distribution of $s_{i}$ conditioned on $z_{i}$ and $\widehat{\beta}_{i}$ (estimated $\beta_{i}$) and features with
\begin{align}
\label{eq:pos}
\mathrm{p}(s_{i} = 0 | z_{i}, \widehat{\beta}_{i}) = \frac{(1 - \widehat{c}_{i})f_{0}(z_{i})}{\widehat{c}_{i}f_{1}(z_{i}) + (1 - \widehat{c}_{i})f_{0}(z_{i})} < \gamma \ \left( \widehat{c}_{i} = e^{\widehat{\beta}_{i}} / \left(1 + e^{\widehat{\beta}_{i}} \right) \right) 
\end{align}
are selected. The $\gamma \in (0,1)$ is pre-setting threshold parameter.

\noindent
\textbf{Heterogenous Spatial Smoothing.} 
However, \eqref{eq-z} may lose spatial structure of non-nulls, especially lesion features. Besides, note that the procedural bias and lesion features are heterogenous in terms of volumetric change and level of spatial coherence. Hence, to capture the spatial structure of heterogenous features, we split the graph of voxels which denotes as $\bm{G}$ \footnote{Here $\bm{G} = (\bm{V},\bm{E})$, where $\bm{V}$ is the node set of voxels, $\bm{E}$ is the edge set of voxel pairs in neighbor (e.g. 3-by-3-by-3).} into three subgraphs, i.e. $\bm{G} = \bm{G}_{1} \cup \bm{G}_{2} \cup \bm{G}_3$ with: 
\begin{subequations}
 \label{eq:subgraph}
	\begin{align}
	\label{eq:subgraph-1}
	\bm{G}_{1} & = \left(\bm{V}_{1}, \bm{E}_{1}\right), \ \bm{V}_{1} = \{i: z_{i} \leq 0\}, \ \bm{E}	_{1}  =  \{(i,j) \in \bm{E}: z_{i} \leq 0, z_{j} \leq 0\} \\
	\label{eq:subgraph-2}
	\bm{G}_{2} & =  \left(\bm{V}_{2}, \bm{E}_{2}\right), \ \bm{V}_{2}  = \{i: z_{i} > 0 \}, \ \bm{E}_{2}  = \{(i,j) \in \bm{E}: z_{i} > 0, z_{j} > 0\} \\
	\label{eq:subgraph-3}
	\bm{G}_{3} & = \left(\bm{V}_{3}, \bm{E}_{3}\right), \ \bm{V}_{3} = \bm{V}_{1} \cup \bm{V}_{2}, \ \bm{E}_{3} = \{(i,j) \in \bm{E}: z_{i} > 0, z_{j} \leq 0\}
	\end{align}
\end{subequations}
where $\bm{G}_{1}$ denotes the subgraph restricted on enlarged voxels (procedural bias since -1 denotes AD); $\bm{G}_{2}$ denotes the subgraph restricted on degenerate voxels (lesion features); $\bm{G}_{3}$ denotes the bipartite graph with the edges connecting enlarged and degenerate voxels. The optimization function can be redefined as:
\begin{align}
g(\beta) = \ell(\beta) + \lambda_{pro} \Vert D_{\bm{G}_{1}} \beta \Vert_{1} + \lambda_{les} \Vert D_{\bm{G}_{2}} \beta \Vert_{1} + \lambda_{pro\text{-}les} \Vert D_{\bm{G}_{3}} \beta \Vert_{1} 
\label{eq:loss-first-penalty-heter}
\end{align}
where $D_{\bm{G}_{k}}\beta = \sum_{(i,j) \in \bm{E}_{k}} \beta_{i} - \beta_{j}$ for $k \in \{1,2,3\}$ denote graph difference operator on $\bm{G}_{k=1,2,3}$. By setting the group of regularization hyper-parameters $\{\lambda_{pro}, \lambda_{les}, \lambda_{pro\text{-}les}\}$ with different values, we can enforce spatial smoothness on three subgraphs at different level in a contrast to the traditional homogeneous regularization in \cite{tansey2017false}. 
The choice of each hyper-parameter, similar to \cite{tansey2017false}, it is a trade-off between over-fitting and over-smoothing. Too small value tends to select features more than needed, while too large value will oversmooth  hence the features are less clustered. Note that lesion features are more spatially coherent than procedural bias and they are located in different regions, the reasonable choice of regularization hyper-parameters tend to have $\lambda_{les} \leq \lambda_{pro} \leq\lambda_{pro\text{-}les}$.\newline
 \noindent \textbf{Optimization.} Note that the function~\eqref{eq:loss-first-penalty-heter} is not convex. Hence we adopted the same idea in \cite{tansey2017false} that introduced the latent variables $s_{i}$ and $=1$ if $z_{i} \sim f_{1}(z)$ and 0 if $z_{i} \sim f_{0}(z)$.
The $\ell(\beta)$ and $g(\beta)$ are modified as:
\begin{align}
\label{eq:loss-h}
\ell(\beta,s) & = \sum_{i=1}^{N} \left\{\log \left(1 + e^{\beta_{i}} \right) - s_{i}\beta_{i}\right\} \\
\label{eq:loss-h-penalty-heter}
g(\beta,s) & = \ell(\beta,s) +  \lambda_{pro} \Vert D_{\bm{G}_{1}} \beta \Vert_{1} + \lambda_{les} \Vert D_{\bm{G}_{2}} \beta \Vert_{1} + \lambda_{pro\text{-}les} \Vert D_{\bm{G}_{3}} \beta \Vert_{1}
\end{align}
To solve~\eqref{eq:loss-h-penalty-heter}, we can implement Expectation-Maximization (EM) algorithm to alternatively solve $\beta$ and $s$. Suppose currently we are in the $(k+1)^{th}$ iteration. \textbf{\emph{In the E-step}}, we can estimate $s_{i}$ by expectation value conditional on $(\beta^k,z_{i})$: $\tilde{s}_{i} = \mathrm{E}(s_{i} | \beta^k, z_{i}) = \frac{c^k_{i} f_{1}(z_{i})}{c^k_{i} f_{1}(z_{i}) + (1 - c^k_{i})f_{0}(z_{i})}$.\newline
\textbf{\emph{In the M-step}}, we plug $\tilde{s}_{i}$ into~\eqref{eq:loss-h-penalty-heter}, denote $\widetilde{D}_{\bm{G}} = \left[D_{\bm{G}^T_{1}}, \frac{\lambda_{les}}{\lambda_{pro}}D_{\bm{G}^T_{2}}, \frac{\lambda_{pro\text{-}les}}{\lambda_{pro}}D_{\bm{G}^T_{3}}\right]^T$ and expand $\ell(\beta | \tilde{s}^k)$ using a second-order Taylor approximation at the $\beta^k$. Then the M-step turns into a generalized lasso problem with square loss:
\begin{align}
\min_{\beta} \frac{1}{2} \Vert \tilde{y} - \widetilde{X}\beta \Vert_{2}^2 + \lambda_{pro} \Vert \widetilde{D}_{\bm{G}} \beta \Vert_{1} 
\label{eq:m-step-leastsquare-new}
\end{align}
where $\widetilde{X} = diag\{\sqrt{w_{1}},...,\sqrt{w_{p}}\}$ and $\tilde{y}_{i} = \sqrt{w_{i}}\left( \beta^k_{i} - \triangledown_{\beta} \ell(\beta | \tilde{s}^k_i)_{|_{\beta^k}} / w_{i}\right)$ with $w_{i} = \triangledown_{\beta}^2 \ell(\beta | \tilde{s}_i)_{|_{\beta^k}}$.
Note that $X$ and $\widetilde{D}_{\bm{G}}$ are sparse matrices, hence~\eqref{eq:m-step-leastsquare-new} can be efficiently solved by Alternating Direction Method of Multipliers (ADMM) \cite{boyd2011distributed} which has a  complexity of $O(p\log{p})$.

\noindent
\textbf{Estimation of $f_{0}$ and $f_{1}$.} Before the iteration, we need to estimate $f_{0}(z)$ and $f_{1}(z)$. The marginal distribution of $z$ can be regarded as mixture models with $p$ components: $z \sim \frac{1}{p} \sum_{i=1}^{p}
g_{i}(z), \ g_{i}(z) = p(s_{i})p(z | s_{i}) = c_{i}f_{1}(z) + (1 - c_{i})f_{0}(z)$
Hence, the marginal distribution of $z$ is $f(z) = \bar{c}f_{1}(z) + (1 - \bar{c})f_{0}(z)$, 
which is equivalent to LocalFDR \cite{efron2016computer}. We can therefore implement the CM (Central Matching) \cite{efron2016computer} method to estimate $\{f_{0}(z),\bar{c}\}$ and kernel density to estimate $f(z)$. The $f_{1}(z)$ can thus be given as $\left(f(z) - f_{0}(z)\bar{c}\right) / (1 - \bar{c})$.

\section{Experimental Results}
\label{sec:experimental}
In this section, we evaluate the proposed method by applying it on the ADNI database \url{http://adni.loni.ucla.edu}. The database is split into 1.5T and 3.0T (namely 15 and 30) MRI scanner magnetic field strength datasets. The 15 dataset contains 64 AD, 110 MCI (Mild Cognitive Impairment) and 90 NC, while the 30 dataset contains 66 AD and 110 NC.
After applying DARTEL VBM \cite{vbm_compute} preprocessing pipeline on the data with scale of 4$\times$4$\times$4 mm$^{3}$ voxel size, there are in total 20,091 voxels with average values in GM population on template greater than 0.1 and they are served as input features. We designed experiments on 1.5T AD/NC, 1.5T MCI/NC and 3.0T AD/NC tasks, namely 15ADNC, 15MCINC and 30ADNC, respectively. 

\subsection{Prediction Results}
\label{prediction}

To test the efficacy of selected features by FDR-HS and compare it with other univariate models (as listed in Table~\ref{table:acc}), we feed them into elastic net classifier, which has been one of the state-of-the-arts in the prediction of neuroimage data \cite{shen2011identifying}. 
The hyper-parameters are determined by grid-search. In details, the threshold hyper-parameter of p-value in T-test and q-value in BH$_{q}$ are optimized through $\{0.001,0.01,0.02,0.05,0.1\}$; the threshold hyper-parameter for choosing non-nulls, i.e. $\gamma$ for FDR-HS~\eqref{eq:pos} and the counterpart of LocalFDR \cite{efron2016computer}, are chosen from $\{0.1,0.2,...,0.5\}$. Besides, the regularization parameters $\lambda_{pro}$, $\lambda_{les}$ and $\lambda_{pro\text{-}les}$ of FDR-HS are ranged in $\{0.1,0.2,...,2\}$. For elastic net, the regularization parameter is chosen from $\{0.1,0.2,...,2,5,10\}$; the mixture parameter $\alpha$ is from $\{0,0.01,...,1\}$. Moreover, we compare our model to GSplit LBI and elastic net, adopting the same optimized strategy for hyper-parameters in \cite{sun2017gsplit} (the top 300 negative voxels are identified as procedural bias \cite{sun2017gsplit}) and those of elastic net following after the univariate models, as mentioned above. 

A 10-fold cross-validation strategy is applied and the classification results for all tasks are summarized in Table~\ref{table:acc}. As shown, our method yields better results than others in all cases, that includes: (1) FDR-HS can select features with more prediction power than other univariate models due to the ability to capture heterogenous spatial patterns; (2) FDR-HS can achieve better classification results than multivariate methods in high dimensional settings, in which the non-nulls may be represented by other nulls that are highly correlated with them. 

\begin{table}[!h]
\caption{Comparison between FDR-HS and others on 10-fold classification result}
\footnotesize
\centering
\begin{tabular}{|c|c|c|c|c|c|c|c|c|c|c|c|c|c|c|c|c|c|}
\hline
\multirow{2}{*}{} & \multicolumn{4}{c|}{Univariate + ElasticNet} &  \multicolumn{2}{c|}{Multivariate} \\ \cline{2-7}
& T-test & BH$_{q}$ \cite{benjamini1995controlling}  & LocalFDR \cite{efron2016computer}  & \textbf{FDR-HS} & GSplit LBI \cite{sun2017gsplit} & Elastic Net \cite{elasticnet} \\
\hline 
15ADNC  & $89.61\%$ &  $89.61\%$ & $87.01\%$ &  $\textbf{90.26\%}$ & $85.06\%$ & $87.01\%$ \\
\hline
15MCINC  & $70.50\%$ & $71.00\%$ & $73.50\%$ &  $\textbf{75.00\%}$ & $72.50\%$ & $72.00\%$ \\
\hline
30ADNC  & $88.64\%$  & $89.77\%$ & $89.77\%$  & $\textbf{91.48\%}$ & $89.77\%$ & $88.07\%$ \\
\hline
\end{tabular} 
\label{table:acc}
\end{table}


\subsection{Feature Selection Analysis}
\label{lesion}
We used 2-d images of 30ADNC to visualize the features of all methods under the hyper-parameters that give the best accuracy. As shown in Fig.~\ref{figure:30ADNC}, the lesion features selected by FDR-HS are located clustered in early damaged regions; while procedural bias are surrounding around lateral ventricle. Besides, such a result is given by $\lambda_{les} < \lambda_{pro} < \lambda_{pro\text{-}les}$, which agrees with that the larger value results in features with lower level of spatial coherence. 
In contrast, the lesions selected by T-test and BH$_{q}$ are scattered and redundant; some procedural bias around lateral ventricle are missed by BH$_{q}$ and LocalFDR. Moreover, GSplit LBI selected procedural bias on regions with CSF space less enlarged than lateral ventricle; besides, it ignored lesions located in medial temporal lobe, Thalamus and Fusiform etc., which are believed to be the early damaged regions \cite{aggleton2016thalamic,galton2001differing}.
\begin{figure}[!t]
    \centering
    \includegraphics[width = 1\textwidth]{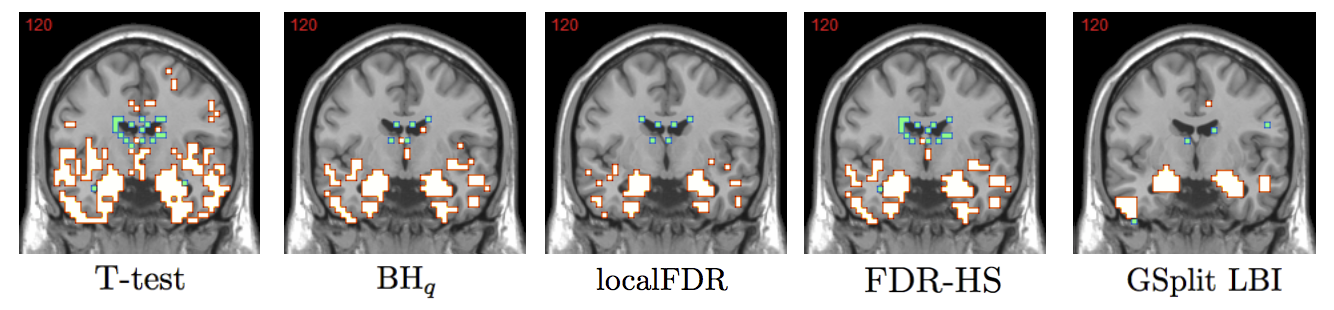}
    \caption{   \footnotesize The comparison of FDR-HS between others in terms of feature selection (30ADNC). Red denotes lesions; blue denotes procedural bias.}
    \label{figure:30ADNC}
\end{figure}


Besides, we also evaluated the stability of selected features using multi-set Dice Coefficient\thinspace(mDC) measurement defined in \cite{n2gfl}. Larger mDC implies more stable feature selection.
As shown in Table~\ref{table:mdc}, our model can obtain more stable results
than GSplit LBI which suffer the ``collinearity" problem. 
\begin{table}
\caption{Comparison between FDR-HS and others on stability (measured by mDC)} 
\footnotesize
\centering
\begin{tabular}{|c|c|c|c|c|c|c|}
\hline
 & T-test & BH$_{q}$ & LocalFDR & FDR-HS & GSplit LBI \\
\hline
mDC$^{(+)}$ (Lesion features) & 0.6705  & 0.6248 & 0.6698 & \textbf{0.6842} & 0.4598 \\
\hline
mDC$^{(-)}$ (Procedural Bias) & 0.6267  & 0.5541 & 0.5127 & \textbf{0.6540} & 0.3033\\
\hline
\end{tabular} 
\label{table:mdc}
\end{table}
\section{Conclusions}
In this paper, a ``two-groups" Empirical-Bayes model is proposed to stably and efficiently select interpretable heterogenous features in voxel-based neuroimage analysis. By modeling prior probability voxel-by-voxel and using a heterogenous regularization, the model can avoid multicollinearity and exploit spatial patterns of features. With experiments on ADNI database, the features selected by our models have better interpretability and prediction power than others. \newline

\setlength{\parindent}{0pt} \textbf{{Acknowledgements.}} This work was supported in part by 973-2015CB351800, NSFC-61625201, 61527804, National Basic Research Program of China (Nos. 2015CB85600, 2012CB825501), NNSF of China (Nos. 61370004, 11421110001), HKRGC grant 16303817, Scientific Research Common Program of Beijing Municipal Commission of Education (No. KM201610025013) and grants from Tencent AI Lab, Si Family Foundation, Baidu BDI and Microsoft Research-Asia.
%
%
\bibliographystyle{splncs03}
\bibliography{fdrhs_refer}

\appendix 
\newpage
\renewcommand\thesection{\Alph{section}}
\begin{center}
    \Large\bf{Supplementary Information}
\end{center}

\section{Multicollinearity problem}

One can implement sparse multivariate models for feature selection and classification by minimizing the penalized optimization function. Particularly, $n^2$GFL \cite{n2gfl} was proposed to stably capture the degenerate voxels by harnessing the sparsity and geometrically clustering properties under the $\beta \geq 0$ constraint and regularization function $\Omega(\beta) = \Vert D\beta \Vert_{1}$ with $D = \left[I; \rho D_{G}\right]$. \footnote{Here $D_{G}\beta = \sum_{(i,j) \in E} \beta_{i} - \beta_{j}$ denotes a graph difference operator on $\bm{G} = (\bm{V},\bm{E})$, where $\bm{V}$ is the node set of voxels, $\bm{E}$ is the edge set of voxel pairs in neighbour (e.g. 3-by-3-by-3).} On the basis of this, GSplit LBI was proposed to capture additional procedural bias via a variable splitting scheme. 

However, such multivariate models suffer from the multicollinearity problem, i.e. high correlation among features, including the following aspects in details: (1) the collinearity between non-nulls and nulls can make the irrepresentable condition of genlasso which ensures the successfully recovery of the true support set \cite{vaiter} hard to satisfy. Such problem can violate the non-nulls to be selected; (2) the collinearity among non-nulls of sparsity model such as lasso/genlasso tends to only select one feature/region among correlated ones. (3) the procedural bias may be represented by other variables which are highly correlated with them due to the ``collinearity" between non-nulls and nulls during minimization of General Linear Model (GLM). Specifically, the general penalized optimization function of the multivariate model is:
\begin{align}
(\beta,\beta_{0}) = \arg_{(\beta,\beta_{0})} \min \ell(\beta,\beta_{0}) + \lambda \Omega(\beta) 
\label{eq:glm}
\end{align}
where $\ell(\beta,\beta_{0})  = \frac{1}{N}\sum_{i=1}^{N} \log\left( 1 + e^{y_{i} \cdot (\beta_{0} + x_{i}^T\beta) } \right) - y_{i} \cdot (\beta_{0} + x_{i}^T\beta)$
with $\beta_{0}$ denoting the bias parameter and $\Omega(\beta)$ is the regularization function. Different choice of $\Omega(\beta)$ leads to different model.

Since GSplit LBI and genlasso enforce same sparsity regularizations on lesion features, we firstly discuss the problem of genlasso under high dimensional space: (1) the irrepresentable condition which ensures the model selection consistency (i.e. successfully recover the true support set) is not easy to satisfy (2) only select one feature/cluster among correlated ones. 

To understand (1), note in \cite{vaiter} that under linear model the necessary condition for genlasso to satisfy model selection consistency is that $IC_{1} < 1$. When $D = I$, the slightly stronger version of this condition is $\Vert X_{S^c}^TX_{S}\left( X_{S}^TX_{S} \right)^{-1} \Vert_{1} < 1$ where $S$ denotes the true support set (lesion voxels). Such irrepresentable condition implies the ``decorrelated" property of $S$ and $S^c$. However, such condition is hard to satisfy under high dimensional space since covariates are more easier to be correlated. 

To see (2), note in \cite{elasticnet} that the lasso only selects one feature among a group of correlated ones. We claim that genlasso also suffers from this limitation. Specifically, it can be shown in Lemma~\ref{lemma:correlated} that the Total Variation regularization of genlasso tends to select only single region among a group of correlated ones. 
\begin{lemma}
\label{lemma:correlated}
Let $\bm{G}_{S_1} = (\bm{V}_{S_{1}},\bm{E}_{S_2})$ and $\bm{G}_{S_2} = (\bm{V}_{S_{2}},\bm{E}_{S_1})$ denotes two subgraphs (regions) that are not connected with other nodes in $\bm{V}$, i.e.
\begin{align}
D_{G} = \left[ \begin{array}{ccc} D_{S_{1}} & & \\ & D_{S_{2}} & \\ & & D_{S_1^c \cap S_2^c} \end{array} \right] \nonumber
\end{align} 
where $S_{1} \subset \{1,...,p\}$, $S_{2} \subset \{1,...,p\}$, $|S_1| = |S_2|$ and $S_{1} \cap S_{2} = \emptyset$. If $D_{S_1} = D_{S_2}$, $X_{S_1} = X_{S_2}$ and $\hat{\beta}$ is the solution of~\ref{eq:glm} with $\Omega{\beta} = \Vert D_{G} \beta \Vert_{1}$, then $\hat{\beta}^{\star}$ is the other solution where 
\begin{align}
\hat{\beta}^{\star}_{k} = 
\begin{cases}
\hat{\beta}_{k} & k \in S_{1}^c \cap S_{2}^c \\
\left(\hat{\beta}_{k} + \hat{\beta}_{k(S_{2})}\right) \cdot s & k \in S_{1} \\
\left(\hat{\beta}_{k(S_1)} + \hat{\beta}_{k}\right) \cdot (1-s) & k \in S_{2}
\end{cases} \nonumber
\end{align}
for any $s \in [0,1]$, where $k(S_{t=1,2})$ are indexes such that $X_{k} = X_{k(S_{t=1,2})}$ for $k \in S_{t=2,1}$. 
\end{lemma}
\begin{proof} 
It can be easily verified from the definition of~\ref{eq:glm}.
\end{proof}

Now we discuss the problem of the procedural bias in GSplit LBI, i.e. the procedural bias may be represented by other variables which are highly correlated with them due to the ``collinearity". Note that to select procedural bias, GSplit LBI adopts an variable splitting scheme $\Vert D\beta - \gamma \Vert_{2}^2$ by introducing an augmented variable $\gamma$, the loss function is redefined as:
\begin{align}
\ell(\beta_{0},\beta,\gamma) = \ell(\beta_{0},\beta) + \frac{1}{2\nu} \Vert D\beta - \gamma \Vert_{2}^2 
\label{eq:split}
\end{align}
and implement it by the following iterative algorithm:
\begin{subequations}
    \label{eq:slbi-show}
    \begin{align}
    	\label{eq:slbi-show-a}
        \beta_0^{k+1} &= \beta_0^k - \kappa \alpha \nabla_{\beta_{0}} \ell(\beta_{0}^k, \beta^k,\gamma^k), \\
        \label{eq:slbi-show-b}
        \beta^{k+1} &= \beta^k - \kappa \alpha \nabla_{\beta} \ell(\beta_0^k,\beta^k,\gamma^k), \\
        \label{eq:slbi-show-c}
        z^{k+1} &= z^k - \alpha \nabla_\gamma \ell(\beta_0^k,\beta^k,\gamma^k), \\
        \label{eq:slbi-show-d}
        \gamma^{k+1}_{\bm{V}} &= \kappa \cdot \max(z_{\bm{V}}^{k+1} - 1,0), \\
        \label{eq:slbi-show-e}
        \gamma^{k+1}_{\bm{G}} &= \kappa \cdot \mathrm{sign} \max(|z_{\bm{G}}^{k+1}| - 1,0), \\
        \label{eq:slbi-show-f}
        \beta_{les}^{k+1} &= (I - D_{S_{k+1}^c}^{\dagger}D_{S_{k+1}^c})\beta^{k+1},
    \end{align}
\end{subequations}
where $S_{k} := \mathrm{supp}(\gamma^k)$. Note that in~\ref{eq:slbi-show-f}, the $\beta_{les}$,  which is the projection of $\beta$ onto the support set of $\gamma$, is the estimator to capture lesion features. The elements with comparably large magnitude among the remainder of such projection are regarded as procedural bias. However, under high dimensional data, such definition may suffer multicollinearity problem that the procedural bias are ``submerged" or represented by other null variables that have high correlation with them. In detail, note that the procedural bias are less clustered than lesion features and also that in~\ref{eq:slbi-show-b}, when $\kappa \to \infty$, $\alpha \to 0$, we have $\beta^{k+1} \to \arg \min_{\beta} \ell(\beta_0^k,\beta,\gamma^{k})$, i.e. minimizing GLM model, it can then be shown in the following lemma that the algorithm may choose null variables while the procedural bias may not be successfully recovered.
\begin{lemma}
Denotes the set of index of procedural bias as $S_{pro}$. Assume $i \in S_{pro}$ and $\{j_{1},...,j_{m}\} \subset S^c_{pro}$ are isolated points satisfying (1) $\gamma^{k}_{\{i,j_{1},...,j_{m}\}} = 0$ and (2) $X_{i} = \sum_{k=1}^{m} X_{j_{k}}s_{k}$ for some $\{s_{1},...,s_{m}\}$ such that $\sum_{k=1}^{m}s_{k} = 1$ and $s_{k} > 0$ for $k \in \{1,..,m\}$. If $\hat{\beta}^{\star}$ minimizes $\ell(\beta_{0}^k,\beta,\gamma^k)$, then $\hat{\beta}^{\star}_{i} = \sum_{k=1}^m \hat{\beta}^{\star}_{j_k} s_{j_k}$, which means $\left| \hat{\beta}^{\star}_{i} \right| \leq \max\{\left| \hat{\beta}^{\star}_{j_{k}} \right|\}_{k=1,...,m}$. Furthermore, if $\left| \hat{\beta}^{\star}_{j_{1}} \right| = \left| \hat{\beta}^{\star}_{j_{2}} \right| = ... = \left| \hat{\beta}^{\star}_{j_{m}} \right|$ does not hold, then $\left| \hat{\beta}^{\star}_{i} \right| < \max\left\{\left| \hat{\beta}^{\star}_{j_{k}} \right|\right\}_{k=1,...,m}$.
\label{lemma:bias}
\end{lemma}

\begin{proof}
Since $\hat{\beta}^{\star}$ minimizes $\ell(\beta^k_{0},\beta,\gamma^k)$, then it's easy to see that
\begin{align}
\hat{\beta}^{\alpha}_{t} = \begin{cases} \alpha \hat{\beta}^{\star}_{i} & t = i \\ (1-\alpha) s_{j_k} \hat{\beta}^{\star}_{i} + \hat{\beta}^{\star}_{s_{j_k}} & t = s_{j_k} (k = 1,...,m) \\ \hat{\beta}^{\star}_{t} & t \notin \{i,j_{1},....,j_{m}\} \end{cases}  \nonumber
\end{align} 
satisfies that $\ell(\beta^k_{0},\hat{\beta}^{\alpha}) = \ell(\beta^k_{0},\hat{\beta}^{\star})$ 
for any $\alpha \in [0,1]$. Denote $\mathcal{I} = \{i,j_{1},...,j_{m}\}$. Since $\mathcal{I}$ corresponds to isolate points and $\gamma^{k}_{\mathcal{I}} = 0$, then 
\begin{align}
\arg \min_{\beta_{\mathcal{I}}}(\beta_{0}^k,\beta_{\mathcal{I}},\hat{\beta}_{\mathcal{I}^c}^{\star},\gamma^k) \iff \arg\min_{\beta_{\mathcal{I}}} \ell(\beta_{0}^k,\beta_{\mathcal{I}},\hat{\beta}_{\mathcal{I}^c}^{\star}) + \frac{1}{2\nu} \Vert \beta_{\mathcal{I}} \Vert_{2}^2 \nonumber
\end{align}
By computing the gradient of $\Vert \hat{\beta}^{\alpha} \Vert_{2}^2$ over $\alpha$, it's easy to see that $\hat{\beta}_{\mathcal{I}}^{\star} = \hat{\beta}^{1}_{\mathcal{I}}$ $ = \arg\min_{\beta_{\mathcal{I}}} \ell(\beta_{0}^k,\beta_{\mathcal{I}},\hat{\beta}_{\mathcal{I}^c}^{\star}) + \frac{1}{2\nu} \Vert \beta_{\mathcal{I}} \Vert_{2}^2$ only if $\hat{\beta}^{\star}_{i} = \sum_{k=1}^m \hat{\beta}^{\star}_{j_k} s_{j_k}$. Since $\left| \cdot \right|$ is an convex function and $\hat{\beta}^{\star}_{i}$ is a convex combination of $\{\hat{\beta}^{\star}_{j_k}\}_{k=1,...,m}$, we then have
\begin{align}
\left| \hat{\beta}^{\star}_{i} \right| \leq \sum_{k=1}^{m} s_{k} \left| \hat{\beta}^{\star}_{j_{k}} \right| \leq \max\left\{\left| \hat{\beta}^{\star}_{j_{k}} \right|\right\}_{k=1,...,m} \nonumber 
\end{align}
The last inequality can be dropped if $\left| \hat{\beta}^{\star}_{j_{1}} \right| = \left| \hat{\beta}^{\star}_{j_{2}} \right| = ... = \left| \hat{\beta}^{\star}_{j_{m}} \right|$ does not hold.

\end{proof}

\section{Heterogenous Features}

We mentioned that the procedural bias and lesion features are heterogenous in terms of volumetric change (enlarged v.s. atrophied) and spatial patterns (surroundingly distributed v.s. spatially cohesive). The heterogeneity in terms of volumetric change is easy to understand, since the procedural bias commonly refer to enlarged GM voxels in voxel-based dementia analysis and lesion features refer to atrophied ones. To illustrate the heterogeneous levels of spatial coherence, we evaluate the edge density in 3D coordinate system by introducing 3D edge density (3dED) measurement for both selected lesions and procedural bias. In detail, $3dED$ is defined as:
\begin{align}
\label{eq:3ded}
3dED^{\pm} = \frac{1}{K} \sum_{k=1}^{K} \frac{|\{(i,j) \in E: i \in S^{\pm}_{k}, j \in S^{\pm}_{k}\}|}{\max_{\widehat{\bm{G}_{3d}}}\left\{\left|\widehat{\bm{E}}\right|: \widehat{\bm{G}_{3d}} = \left(\widehat{\bm{V}},\widehat{\bm{E}}\right), \ \left| \widehat{\bm{V}} \right| = \left| \{i: i \in S^{\pm}_{k}\} \right| \right\}} 
\end{align}
where $S^{\pm}_{k}$ denote the support set of lesion features and procedural bias in the k-th fold, respectively. For each fold $k$, we need to compute:
\begin{equation*}
\max_{\widehat{\bm{G}_{3d}}}\left\{\left|\widehat{\bm{E}}\right|: \widehat{\bm{G}_{3d}} = \left(\widehat{\bm{V}},\widehat{\bm{E}}\right), \ \left| \widehat{\bm{V}} \right| = \left| \{i: i \in S^{\pm}_{k}\} \right| \right\}
\end{equation*}
For the graphs with $p$ voxels which are embedded in 3-d coordinate space, the maximum number of edges is equal to:
\begin{equation*}
\max_{(c_{1},c_{2},c_{3},s_{1},s_{2},r) \in \mathcal{S}^{3d}_{p}}\{3c_{1}c_{2}c_{3} - c_{1}c_{2} - c_{1}c_{3} - c_{2}c_{3} + 2r_{1}r_{2} - r_{1} - r_{2} + l\}
\end{equation*}
where
\begin{align}
\mathcal{S}^{3d}_{p} & = \{\left[c_{1},c_{2},c_{3},r_{1},r_{2},l\right] \in \mathbb{R}_{1\times 6}: c_{1}c_{2}c_{3} + r_{1}r_{2} + l = p, r \leq \max\{r_{1},r_{2}\}, \nonumber \\
& \{r_{1}r_{2} + r \leq c_{1}c_{2}, r_{1}\leq c_{1}, r_{2} \leq c_{2}\} \ \mathrm{or} \ \{r_{1}r_{2} + l \leq c_{1}c_{3}, r_{1}\leq c_{1}, r_{2} \leq c_{3}\}  \nonumber \\
& \mathrm{or} \ \{r_{1}r_{2} + l \leq c_{2}c_{3}, r_{1}\leq c_{2}, r_{2} \leq c_{3}\} \}
  \nonumber 
 \end{align}
where $\mathbb{R}^{+} = \{x \in R, x \geq 0\}$. Besides, the $c_{1},c_{2},c_{3}$ can be taken as length, width and height of a cube, $r_{1}r_{2} + l$ are the remainder of $p$ for a cube $(c_{1},c_{2},c_{3})$. The $r_{1},r_{2}$ can be taken as the length and width of the rectangle which is located on one of the surfaces of $(c_{1},c_{2},c_{3})$. The $l$ is the residual of $p - c_{1}c_{2}c_{3}$ for rectangle $(r_{1},r_{2})$.

According to~\ref{eq:3ded}, the procedural bias are turned to be selected much less clustered than lesion features by all univariate models, e.g. the BH$_{q}$ yields 0.4677 for lesion features and 0.1621 for procedural bias; while localFDR has 0.4662 and 0.1699; FDR-HS has 0.5365 and 0.2535, which validates the heterogenous assumption in terms of the level of spatial coherence.

\section{IDS of ADNI subject used in our experiments}



\small
\begin{center}
\begin{longtable}{c|c|c||c|c|c||c|c|c}
 \toprule
 Subject & ID & Class &  Subject & ID & Class &  Subject & ID & Class  \\
 \midrule
 
  123\textunderscore S\textunderscore0094	 & 	9655	 & 	15AD   &  
137\textunderscore S\textunderscore 0158        &       11127      &      15MCI &  
014\textunderscore S\textunderscore  0558  &     17400          &       15NC \\

123\textunderscore S\textunderscore 0088	 & 	9788	 & 	15AD   &  
128\textunderscore S\textunderscore 0225         &    11179        &     15MCI & 
 021\textunderscore S\textunderscore 0647   &    17668           &       15NC \\

098\textunderscore S\textunderscore 0149	 & 	10146	 & 	15AD  &  136\textunderscore S\textunderscore 0107         &     11227       &       15MCI	 &  
 137 \textunderscore S\textunderscore  0686  &  17813             &       15NC \\

032\textunderscore S\textunderscore 0147	 & 	10404	 & 	15AD  &  032\textunderscore S\textunderscore 0214      &       11280        &       15MCI &    032\textunderscore S\textunderscore 0677   &   17820            &       15NC \\

123\textunderscore S\textunderscore 0162  & 	10962	 & 	15AD   &
005\textunderscore S\textunderscore 0222 &      11299        &       15MCI &
002\textunderscore S\textunderscore  0685  &  18211             &       15NC \\

128\textunderscore S\textunderscore 0216	 & 	11101	 & 	15AD  & 027\textunderscore S\textunderscore 0179  &        11348       &       15MCI &
 094\textunderscore S\textunderscore  0711  &   18589            &       15NC \\
 
128\textunderscore S\textunderscore 0167	 & 	11203	 & 	15AD & 
021\textunderscore S\textunderscore  0231  &        11430       &       15MCI	&
 127\textunderscore S\textunderscore 0684   &    18896           &       15NC \\
 
005\textunderscore S\textunderscore 0221	 & 	11604	 & 	15AD &
007\textunderscore S\textunderscore 0249 &         11544      &       15MCI	&
  033\textunderscore S\textunderscore  0734  &     19155          &       15NC  \\
  
 014\textunderscore S\textunderscore 0328	 & 	12327	 & 	15AD &
098\textunderscore S\textunderscore 0269  &        11615       &       15MCI & 
 033\textunderscore S\textunderscore 0741   &   19258            &       15NC \\
 
007\textunderscore S\textunderscore 0316	 & 	12616	 & 	15AD &
130\textunderscore S\textunderscore 0289 &         11850      &       15MCI	&
  094\textunderscore S\textunderscore  0692  &     19567          &       15NC \\
 
021\textunderscore S\textunderscore 0343	 & 	12979	 & 	15AD & 
021\textunderscore S\textunderscore  0273 &         11942     &       15MCI	&
 009 \textunderscore S\textunderscore  0751  &      20013         &       15NC \\

014\textunderscore S\textunderscore 0356	 & 	13004	 & 	15AD &
007\textunderscore S\textunderscore 0293  &         11982       &       15MCI	&
116\textunderscore S\textunderscore  0648  &     20370          &       15NC \\

032\textunderscore S\textunderscore 0400	 & 	13525	 & 	15AD & 
031\textunderscore S\textunderscore 0294  &        12065      &       15MCI &
 129\textunderscore S\textunderscore 0778   &      20543         &       15NC \\
 
116\textunderscore S\textunderscore 0370	 & 	14122	 & 	15AD & 
021\textunderscore S\textunderscore 0276  &        12092      &       15MCI & 
 029\textunderscore S\textunderscore 0824   &     23213          &       15NC \\
 
127\textunderscore S\textunderscore 0431	 & 	15497	 & 	15AD & 
128\textunderscore S\textunderscore 0227 &         12119      &       15MCI &  
116\textunderscore S\textunderscore  0657  &     23350          &       15NC \\

031\textunderscore S\textunderscore 0554	 & 	15994	 & 	15AD & 
027\textunderscore S\textunderscore 0256  &        12250       &       15MCI &
006\textunderscore S\textunderscore  0731  &    23468           &       15NC \\

128\textunderscore S\textunderscore 0517	 & 	16150	 & 	15AD &
130\textunderscore S\textunderscore 0285   &      12424         &       15MCI & 
029\textunderscore S\textunderscore  0845  &     24249          &       15NC\\

 116\textunderscore S\textunderscore 0487	 & 	16377	 & 	15AD & 
098\textunderscore S\textunderscore 0288   &      12654        &       15MCI & 
 009\textunderscore S\textunderscore  0862  &     25128          &       15NC \\

 002\textunderscore S\textunderscore 0619	 & 	16392	 & 	15AD & 
007\textunderscore S\textunderscore 0344   &      12697        &       15MCI	&  
098\textunderscore S\textunderscore  0896  &    25255           &       15NC \\

131\textunderscore S\textunderscore 0497	 & 	16666	 & 	15AD & 
021\textunderscore S\textunderscore 0332   &     12862          &       15MCI	&  033\textunderscore S\textunderscore  0923  &    25427           &       15NC \\

021\textunderscore S\textunderscore 0642	 & 	17632	 & 	15AD & 
128\textunderscore S\textunderscore 0258  &      13085        &       15MCI  &
130\textunderscore S\textunderscore  0886  &      25455         &       15NC  \\

033\textunderscore S\textunderscore 0739	 & 	19175	 & 	15AD & 
027\textunderscore S\textunderscore 0307  &      13281         &       15MCI	&
006\textunderscore S\textunderscore 0498   &   25790            &       15NC \\

100\textunderscore S\textunderscore 0743	 & 	19585	 & 	15AD & 
123\textunderscore S\textunderscore  0390   &     13315          &       15MCI	&
052\textunderscore S\textunderscore  0951  &   26642            &       15NC \\

033\textunderscore S\textunderscore 0724	 & 	19772	 & 	15AD &
031\textunderscore S\textunderscore 0351  &     13783          &       15MCI	&
130\textunderscore S\textunderscore  0969  &    26688           &       15NC \\

128\textunderscore S\textunderscore 0740	 & 	19990	 & 	15AD & 
021\textunderscore S\textunderscore  0424  &     13909         &       15MCI	&
021\textunderscore S\textunderscore  0984  &     27056          &       15NC \\

021\textunderscore S\textunderscore 0753	 & 	20169	 & 	15AD & 
053\textunderscore S\textunderscore  0389  &     13938          &       15MCI	&
024\textunderscore S\textunderscore  0985  &     27607          &       15NC \\

137\textunderscore S\textunderscore 0796	 & 	23112	 & 	15AD &
094\textunderscore S\textunderscore  0434  &   13964           &       15MCI	&  
024\textunderscore S\textunderscore  1063  &     28111          &       15NC \\

029\textunderscore S\textunderscore 0836	 & 	23231	 & 	15AD &
068\textunderscore S\textunderscore  0401  &     14161          &       15MCI	&  033\textunderscore S\textunderscore  1098  &    30304           &       15NC  \\

100\textunderscore S\textunderscore 0747	 & 	23581	 & 	15AD &
131\textunderscore S\textunderscore 0409  &      14240        &       15MCI	&
  010\textunderscore S\textunderscore  0472  &     30481          &       15NC \\

127\textunderscore S\textunderscore 0754	 & 	23787	 & 	15AD &
116\textunderscore S\textunderscore 0361  &       14296       &       15MCI	&
 137\textunderscore S\textunderscore 0972   &     31702          &       15NC \\

012\textunderscore S\textunderscore 0803	 & 	24863	 & 	15AD & 
132\textunderscore S\textunderscore 0339  &       14367        &       15MCI	& 
  033\textunderscore S\textunderscore  1086  &     32054          &       15NC \\

033\textunderscore S\textunderscore 0889	 & 	25026	 & 	15AD &
037\textunderscore S\textunderscore 0377  &       14405       &       15MCI	 &
130\textunderscore S\textunderscore  1200  &     36281          &       15NC \\

126\textunderscore S\textunderscore 0891	 & 	25172	 & 	15AD & 
027\textunderscore S\textunderscore  0485  &      14928         &       15MCI	& 
116\textunderscore S\textunderscore 1232   &     37848          &       15NC \\

005\textunderscore S\textunderscore 0929	 & 	25645	 & 	15AD & 
027\textunderscore S\textunderscore  0408  &     14964          &       15MCI	&  027\textunderscore S\textunderscore  0120  &   10933            &       15NC \\

006\textunderscore S\textunderscore 0547	 & 	25816	 & 	15AD & 
137\textunderscore S\textunderscore 0481  &    15044          &       15MCI	 &
 068\textunderscore S\textunderscore  0127  &   11133            &       15NC \\

002\textunderscore S\textunderscore 0955	 & 	26170	 & 	15AD & 
027\textunderscore S\textunderscore  0417  &    15148          &       15MCI	&
068\textunderscore S\textunderscore 0210   &      11235         &       15NC \\

130\textunderscore S\textunderscore 0956	 & 	27032	 & 	15AD & 
053\textunderscore S\textunderscore  0507 &    15315           &       15MCI	&
136\textunderscore S\textunderscore  0186  &     11335          &       15NC \\

053\textunderscore S\textunderscore 1044	 & 	27782	 & 	15AD & 
094\textunderscore S\textunderscore  0531  &     15431          &       15MCI	&
009\textunderscore S\textunderscore 0842 & 24339 & 15NC  \\

133\textunderscore S\textunderscore 1055	 & 	29381	 & 	15AD & 
033\textunderscore S\textunderscore  0567  &      15459         &       15MCI	&   029\textunderscore S\textunderscore 0843& 24406 & 15NC \\

100\textunderscore S\textunderscore 1062	 & 	29579	 & 	15AD & 
127\textunderscore S\textunderscore  0394 &    15510           &       15MCI &
 032\textunderscore S\textunderscore 1169 & 34067 & 15NC \\

029\textunderscore S\textunderscore 1056	 & 	30618	 & 	15AD & 
033\textunderscore S\textunderscore  0514  &     15605          &       15MCI &
018\textunderscore S\textunderscore 0055 &  9136 & 15NC	 \\ 

029\textunderscore S\textunderscore 0999	 & 	31239	 & 	15AD & 
033\textunderscore S\textunderscore  0513  &   15622           &       15MCI	&
100\textunderscore S\textunderscore   0015 &      8390         &       30NC  \\

006\textunderscore S\textunderscore 0653	 & 	31252	 & 	15AD & 
130\textunderscore S\textunderscore  0460  &   15711            &       15MCI	& 136\textunderscore S\textunderscore  0196  &     14236          &       30NC \\

014\textunderscore S\textunderscore 1095	 & 	31576	 & 	15AD & 
098\textunderscore S\textunderscore  0542  &   15848            &       15MCI &
136\textunderscore S\textunderscore   0086 &     14712          &       30NC \\

094\textunderscore S\textunderscore 1090	 & 	31678	 & 	15AD &
007\textunderscore S\textunderscore 0414 &   15875           &       15MCI  &
018\textunderscore S\textunderscore  0369  &       15110        &       30NC \\

021\textunderscore S\textunderscore 1109	 & 	31784	 & 	15AD & 
031\textunderscore S\textunderscore  0568 &   15885            &       15MCI	&
131\textunderscore S\textunderscore  0441  &     15959          &       30NC \\

024\textunderscore S\textunderscore 1171	 & 	35190	 & 	15AD & 
037\textunderscore S\textunderscore  0501  &    15916           &       15MCI  & 
032\textunderscore S\textunderscore  0479  &   16652            &       30NC \\

133\textunderscore S\textunderscore 1170	 & 	35211	 & 	15AD & 
037\textunderscore S\textunderscore  0552  &    15970           &       15MCI	& 018\textunderscore S\textunderscore 0425   &     17168          &       30NC \\

031\textunderscore S\textunderscore 1209	 & 	36178	 & 	15AD & 
130\textunderscore S\textunderscore  0423  &    16196           &       15MCI	&
 126\textunderscore S\textunderscore 0405   &     17177          &       30NC \\

130\textunderscore S\textunderscore 1201	 & 	36269	 & 	15AD & 
014\textunderscore S\textunderscore  0557  &   16304            &       15MCI &
005\textunderscore S\textunderscore  0553  &      17619         &       30NC \\

027\textunderscore S\textunderscore 1081	 & 	37145	 & 	15AD & 
033\textunderscore S\textunderscore 0511   &   16314            &       15MCI	&  
126\textunderscore S\textunderscore  0605  &    17639           &       30NC \\

126\textunderscore S\textunderscore 1221	 & 	37339	 & 	15AD & 
130\textunderscore S\textunderscore 0449   &  16351             &       15MCI	&
005\textunderscore S\textunderscore 0602   &   19615            &       30NC \\

029\textunderscore S\textunderscore 1184	 & 	37350	 & 	15AD & 
027\textunderscore S\textunderscore 0461   &  16467           &       15MCI	 &      012\textunderscore S\textunderscore  1009  &    28962           &       30NC \\

027\textunderscore S\textunderscore 1254	 & 	37859	 & 	15AD &  128\textunderscore S\textunderscore 0608   &  16503           &       15MCI	 &  
012\textunderscore S\textunderscore  1212  &   37403            &       30NC \\

130\textunderscore S\textunderscore 1290	 & 	38395	 & 	15AD & 
128\textunderscore S\textunderscore 0611   & 16766              &       15MCI	 &  007\textunderscore S\textunderscore 1206   & 37761              &       30NC  \\

033\textunderscore S\textunderscore 1285	 & 	38593	 & 	15AD & 
053\textunderscore S\textunderscore 0621   &  16864            &       15MCI	& 
 068\textunderscore S\textunderscore 1191   &   38370            &       30NC \\

033\textunderscore S\textunderscore 1283	 & 	38617	 & 	15AD & 
037\textunderscore S\textunderscore  0566  & 16886             &       15MCI	& 
007\textunderscore S\textunderscore  1222  &   38482            &       30NC \\

033\textunderscore S\textunderscore 1308	 & 	40114	 & 	15AD & 
037\textunderscore S\textunderscore 0539  & 17018              &       15MCI	&
 094\textunderscore S\textunderscore 1241   &    41449           &       30NC \\

024\textunderscore S\textunderscore 1307	 & 	41527	 & 	15AD &
137\textunderscore S\textunderscore 0443  & 17030             &       15MCI	 &
002\textunderscore S\textunderscore  1261  &     41799          &       30NC \\

007\textunderscore S\textunderscore 1339	 & 	42344	 & 	15AD &  005\textunderscore S\textunderscore  0546    & 17056              &       15MCI &  002\textunderscore S\textunderscore 1280   &      41806         &       30NC  \\

130\textunderscore S\textunderscore 1337	 & 	42930	 & 	15AD & 
137\textunderscore S\textunderscore 0631  &    17109          &       15MCI	 &
 052\textunderscore S\textunderscore  1251  &      43812         &       30NC \\

127\textunderscore S\textunderscore 1382	 & 	45060	 & 	15AD & 
027\textunderscore S\textunderscore 0644   &   17157           &       15MCI  &
100\textunderscore S\textunderscore 1286   &      45761         &       30NC \\

094\textunderscore S\textunderscore 1397	 & 	51790	 & 	15AD & 
133\textunderscore S\textunderscore 0629  &     17596          &       15MCI	 &   094\textunderscore S\textunderscore 1267   &    46457           &       30NC \\

094\textunderscore S\textunderscore 1402	 & 	54220	 & 	15AD &
021\textunderscore S\textunderscore 0626   &   17687            &       15MCI	 &  131\textunderscore S\textunderscore  1301  &     49328          &       30NC \\

136\textunderscore S\textunderscore 0299        &      15181       &      30AD & 
098\textunderscore S\textunderscore 0667  &  17702             &       15MCI & 
098\textunderscore S\textunderscore  4003  &     224603          &       30NC  \\

136\textunderscore S\textunderscore 0426        &      16172       &      30AD &
052\textunderscore S\textunderscore 0671  &    17849            &       15MCI & 
098\textunderscore S\textunderscore  4018  &    228788           &       30NC \\

018\textunderscore S\textunderscore 0335        &      16560       &      30AD & 
014\textunderscore S\textunderscore 0563   &    17876           &       15MCI &    031\textunderscore S\textunderscore  4021  &    229148           &       30NC\\

136\textunderscore S\textunderscore 0300        &      16719       &      30AD & 
007\textunderscore S\textunderscore 0698  &    18363           &       15MCI & 
012\textunderscore S\textunderscore  4026  &     238532          &       30NC\\

018\textunderscore S\textunderscore 0633        &      19093       &      30AD &
133\textunderscore S\textunderscore 0638   &    18672           &       15MCI &
098\textunderscore S\textunderscore  4050  &     238615          &       30NC  \\

012\textunderscore S\textunderscore 0689        &      19210       &      30AD &  033\textunderscore S\textunderscore 0723  &     19014          &       15MCI &    016\textunderscore S\textunderscore 4097  &     243556          &       30NC \\

126\textunderscore S\textunderscore 0606        &      20487       &      30AD & 
032\textunderscore S\textunderscore  0718  &     19035          &       15MCI &
 016\textunderscore S\textunderscore 4952   &    337793           &       30NC \\

131\textunderscore S\textunderscore 0691        &      20681       &      30AD &
126 \textunderscore S\textunderscore 0708  &     19089         &       15MCI & 
016\textunderscore S\textunderscore 4121   &      246002         &       30NC  \\

005\textunderscore S\textunderscore 0814        &      24734       &      30AD &
128\textunderscore S\textunderscore  0715  &     19225          &       15MCI & 
 006\textunderscore S\textunderscore 4150   &    249403           &       30NC  \\

002\textunderscore S\textunderscore 0816        &      25405       &      30AD & 
033\textunderscore S\textunderscore  0725  &    19404           &       15MCI & 
 127\textunderscore S\textunderscore  4148  &    250137           &       30NC \\

127\textunderscore S\textunderscore 0844        &      29230       &      30AD & 
137\textunderscore S\textunderscore 0669  &     19419          &       15MCI & 
   003\textunderscore S\textunderscore  4119  &    250894           &       30NC \\

002\textunderscore S\textunderscore 1018        &      33832       &      30AD & 
116\textunderscore S\textunderscore 0649  &    19516           &       15MCI & 
127\textunderscore S\textunderscore  4198  &      254320         &       30NC\\

031\textunderscore S\textunderscore 4024        &    228879       &      30AD &
130\textunderscore S\textunderscore 0505  &   19701            &       15MCI &
 002\textunderscore S\textunderscore  4213  &     254582          &       30NC \\

016\textunderscore S\textunderscore 4009        &    240946       &      30AD &
137\textunderscore S\textunderscore 0722  &    19707           &       15MCI & 
 031\textunderscore S\textunderscore  4218  &    255978           &       30NC \\

094\textunderscore S\textunderscore 4089        &    242719       &      30AD &  126\textunderscore S\textunderscore  0709   &    19754           &       15MCI & 002\textunderscore S\textunderscore 4225   &    257270           &       30NC\\

006\textunderscore S\textunderscore 4153        &    248517       &      30AD & 
128\textunderscore S\textunderscore 0770  &    19907           &       15MCI &
 002\textunderscore S\textunderscore  4262  &   259653            &       30NC \\

003\textunderscore S\textunderscore 4136        &    250173       &      30AD & 
014\textunderscore S\textunderscore 0658  &   20003            &       15MCI & 
 941\textunderscore S\textunderscore  4100  &   259781            &       30NC\\

003\textunderscore S\textunderscore 4152        &    253760       &      30AD & 
137\textunderscore S\textunderscore  0668  &    20202           &       15MCI &
 002\textunderscore S\textunderscore 4264   &    259796           &       30NC   \\

098\textunderscore S\textunderscore 4215        &    255843       &      30AD &
137\textunderscore S\textunderscore 0800  &   20500            &       15MCI &
  021\textunderscore S\textunderscore  4276  &    260047           &       30NC \\

098\textunderscore S\textunderscore 4201        &    256178       &      30AD & 
002\textunderscore S\textunderscore  0782  &    20519           &       15MCI &  
 029\textunderscore S\textunderscore  4290  &    260425           &       30NC\\

006\textunderscore S\textunderscore 4192        &    258594       &      30AD &  130\textunderscore S\textunderscore   0783  &   20794            &       15MCI &   098\textunderscore S\textunderscore  4275   &   261459            &       30NC \\

019\textunderscore S\textunderscore 4252        &    258947       &      30AD &  116\textunderscore S\textunderscore  0752 &     23097          &       15MCI & 
094\textunderscore S\textunderscore 4234   &    261531           &       30NC \\

024\textunderscore S\textunderscore 4280        &    261332       &      30AD & 
068\textunderscore S\textunderscore  0802  &     23389          &       15MCI &  018\textunderscore S\textunderscore  4257  &     262076          &       30NC \\

094\textunderscore S\textunderscore 4282        &    261855       &      30AD & 
133\textunderscore S\textunderscore  0792  &      23444         &       15MCI &     136\textunderscore S\textunderscore  4269  &    264215           &       30NC \\

029\textunderscore S\textunderscore 4307        &    267595       &      30AD &  006\textunderscore S\textunderscore  0675  &   23644            &       15MCI &    029\textunderscore S\textunderscore 4279   &   265980            &       30NC \\

016\textunderscore S\textunderscore 4353        &    267937       &      30AD &  031\textunderscore S\textunderscore 0821  &     23658          &       15MCI &     021\textunderscore S\textunderscore  4335  &   266174            &       30NC \\

109\textunderscore S\textunderscore 4378        &    270669       &      30AD &  133\textunderscore S\textunderscore 0771   &    23876           &       15MCI &      130\textunderscore S\textunderscore   4343 &     266217          &       30NC \\

126\textunderscore S\textunderscore 4494        &    281605       &      30AD &  133\textunderscore S\textunderscore  0727   &   23939            &       15MCI &   018\textunderscore S\textunderscore  4349  &     266625          &       30NC \\

127\textunderscore S\textunderscore 4500        &    283515       &      30AD &   027\textunderscore S\textunderscore  0835  &   24138           &       15MCI &
129\textunderscore S\textunderscore  4369  &   267405            &       30NC \\

007\textunderscore S\textunderscore 4568        &    287472       &      30AD &
031\textunderscore S\textunderscore  0830  &   24281            &       15MCI &
 130\textunderscore S\textunderscore   4352 &     267711          &       30NC  \\

006\textunderscore S\textunderscore 4546        &    287994       &      30AD & 
100\textunderscore S\textunderscore  0035  &   8120            &       15NC &  
129\textunderscore S\textunderscore 4371   &     268462          &       30NC\\

130\textunderscore S\textunderscore 4589        &    291219       &      30AD &  
100\textunderscore S\textunderscore 0047   &    8899           &       15NC & 
   018\textunderscore S\textunderscore 4313   &       268930        &       30NC\\

016\textunderscore S\textunderscore 4591        &    292433       &      30AD &
010\textunderscore S\textunderscore 0067   &        9093       &       15NC & 
019\textunderscore S\textunderscore 4367   &      269273         &       30NC   \\

016\textunderscore S\textunderscore 4583        &    294209       &      30AD & 
 018\textunderscore S\textunderscore  0043  &      9324         &       15NC & 
007\textunderscore S\textunderscore 4387   &     269929          &       30NC\\

014\textunderscore S\textunderscore  4615       &    294334       &      30AD &  
100\textunderscore S\textunderscore  0069  &     9417          &       15NC &
036\textunderscore S\textunderscore 4389   &     270462          &       30NC\\

130\textunderscore S\textunderscore 4641        &    295961       &      30AD & 
032\textunderscore S\textunderscore  0095  &     9680          &       15NC &
 003\textunderscore S\textunderscore  4350  &    270999           &       30NC\\

130\textunderscore S\textunderscore 4660        &    300034       &      30AD &
123\textunderscore S\textunderscore 0072   &    9752           &       15NC & 
 129\textunderscore S\textunderscore  4422  &    272184           &       30NC\\

019\textunderscore S\textunderscore 4549        &    300335       &      30AD &  
007\textunderscore S\textunderscore 0070    &     10027          &       15NC &
018\textunderscore S\textunderscore   4399 &    272231           &       30NC \\

126\textunderscore S\textunderscore 4686        &    300818       &      30AD &  
131\textunderscore S\textunderscore   0123 &    10043           &       15NC & 
018\textunderscore S\textunderscore   4399 &    272231           &       30NC \\

005\textunderscore S\textunderscore 4707        &    304663       &      30AD &  
123\textunderscore S\textunderscore  0106  &       10126        &       15NC &    
 021\textunderscore S\textunderscore 4421   &    273564           &       30NC\\

021\textunderscore S\textunderscore 4718        &    304749       &      30AD & 
027\textunderscore S\textunderscore 0118   &        11370       &       15NC & 
029\textunderscore S\textunderscore   4383 &   273993            &       30NC\\
 
018\textunderscore S\textunderscore 4733        &    306069       &      30AD &
 098\textunderscore S\textunderscore 0172   &      11398   &       15NC & 
003\textunderscore S\textunderscore  4441  &    277108           &       30NC \\

130\textunderscore S\textunderscore 4730        &    306384       &      30AD &  
130\textunderscore S\textunderscore   0232 &    11567           &       15NC  & 
 136\textunderscore S\textunderscore  4433  &      278511         &       30NC \\

137\textunderscore S\textunderscore 4756        &    307118       &      30AD &  
 005\textunderscore S\textunderscore  0223  &     11645          &       15NC &
 006\textunderscore S\textunderscore 4449   &    279470           &       30NC\\
 
  027\textunderscore S\textunderscore 4801        &    314034       &      30AD &  
  123\textunderscore S\textunderscore  0113  &     11714          &       15NC  &
031\textunderscore S\textunderscore 4474   &      280369         &       30NC \\

027\textunderscore S\textunderscore 4802        &    317195       &      30AD &
128\textunderscore S\textunderscore  0230  &     11806          &       15NC  & 
 007\textunderscore S\textunderscore 4488   &    281560           &       30NC \\
   
006\textunderscore S\textunderscore 4867        &    322012       &      30AD & 
137\textunderscore S\textunderscore 0283   &    12028           &       15NC &   
 006\textunderscore S\textunderscore 4485   &     281882          &       30NC \\

016\textunderscore S\textunderscore 4887        &    325649       &      30AD & 
128\textunderscore S\textunderscore  0245  &    12242           &       15NC  &   
010\textunderscore S\textunderscore  4345  &   282005            &       30NC\\

 007\textunderscore S\textunderscore 4911        &    328196       &      30AD &  
 128\textunderscore S\textunderscore 0272   &     12313          &       15NC &  
031\textunderscore S\textunderscore  4496  &    282638           &       30NC\\

021\textunderscore S\textunderscore 4924        &    331257       &      30AD & 
128\textunderscore S\textunderscore 0229   &      12459         &       15NC &
 098\textunderscore S\textunderscore  4506  &    282934           &       30NC\\

137\textunderscore S\textunderscore 4756        &    332930       &      30AD & 
021\textunderscore S\textunderscore  0337  &    12466           &       15NC &
 094\textunderscore S\textunderscore  4459   &     283445          &       30NC\\

127\textunderscore S\textunderscore 4940        &    335512       &      30AD & 
 098\textunderscore S\textunderscore   0171 &     10818          &       15NC &
094\textunderscore S\textunderscore 4460   &     283573          &       30NC\\

027\textunderscore S\textunderscore 4938        &    336926       &      30AD & 
072\textunderscore S\textunderscore 0315  &     12559          &       15NC & 
010\textunderscore S\textunderscore 4442   &    283915           &       30NC\\

027\textunderscore S\textunderscore 4962        &    338558       &      30AD &  
137\textunderscore S\textunderscore 0301        &    12584 &     15NC &  
 007\textunderscore S\textunderscore 4516   &    284424           &      30NC  \\

 130\textunderscore S\textunderscore 4982        &    341787       &      30AD &  
 002\textunderscore S\textunderscore  0295        &     13722        &     15NC & 
 029\textunderscore S\textunderscore  4385   &     285589          &       30NC \\

 130\textunderscore S\textunderscore 4984        &    342274       &      30AD & 
 037\textunderscore S\textunderscore 0327         &   13802               &      15NC &
   094\textunderscore S\textunderscore  4503  &     286222          &       30NC\\

130\textunderscore S\textunderscore 4971        &    342338       &      30AD & 
027\textunderscore S\textunderscore  0403  &   14146            &       15NC &
  073\textunderscore S\textunderscore  4559  &    286553           &       30NC \\

127\textunderscore S\textunderscore 4992        &    342697       &      30AD & 
 137\textunderscore S\textunderscore  0459  &    14178           &       15NC  &
021\textunderscore S\textunderscore  4558  &    287527           &       30NC\\

019\textunderscore S\textunderscore 5012        &    343916       &      30AD &   
002\textunderscore S\textunderscore   0413 &     14437          &       15NC&  
109\textunderscore S\textunderscore  4499  &      288999         &       30NC \\
   
019\textunderscore S\textunderscore 5019        &    345663       &      30AD & 
068\textunderscore S\textunderscore 0473    &     14483          &       15NC &
100\textunderscore S\textunderscore 4469   &     289564          &       30NC \\
 
   002\textunderscore S\textunderscore 5018        &    346242       &      30AD & 
116\textunderscore S\textunderscore 0360   &    14623           &       15NC &  
100\textunderscore S\textunderscore  4511  &         289653      &       30NC\\

  127\textunderscore S\textunderscore 5028        &    346696       &      30AD & 
133\textunderscore S\textunderscore  0488  &    14838           &       15NC &
012\textunderscore S\textunderscore  4545  &      290413         &       30NC \\

130\textunderscore S\textunderscore 4997        &    347410       &      30AD &
 133\textunderscore S\textunderscore  0493  &     14848          &       15NC &
053\textunderscore S\textunderscore  4578  &     290814          &       30NC\\

 005\textunderscore S\textunderscore 5038        &    351432       &      30AD &
014\textunderscore S\textunderscore 0520   &    15299           &       15NC & 
  127\textunderscore S\textunderscore  4604  &      291523        &       30NC\\

   127\textunderscore S\textunderscore 5056        &    353203       &      30AD &  
   014 \textunderscore S\textunderscore  0519  &     15323          &       15NC & 
007\textunderscore S\textunderscore   4620 &     293938         &       30NC \\

127\textunderscore S\textunderscore 5058        &    354636       &      30AD & 
 116 \textunderscore S\textunderscore  0382  &      15347         &       15NC & 
127\textunderscore S\textunderscore   4645 &    295590           &       30NC\\

007\textunderscore S\textunderscore 0128        &    10007         &    15MCI &   
  128\textunderscore S\textunderscore  0500  &  15366             &       15NC &
 002\textunderscore S\textunderscore 4270       &      260581      &        30NC \\

010\textunderscore S\textunderscore 0161        &    10077         &    15MCI & 
 010\textunderscore S\textunderscore  0419  &    15415           &       15NC &
013\textunderscore S\textunderscore   4579 &    296776           &       30NC \\

021\textunderscore S\textunderscore 0141        &    10173         &    15MCI &  
131\textunderscore S\textunderscore  0436  &    15674           &       15NC &
 013\textunderscore S\textunderscore  4580  &   296859            &       30NC  \\

127\textunderscore S\textunderscore 0112        &     10419        &    15MCI & 
128\textunderscore S\textunderscore 0522   & 15821              &       15NC &    
 012\textunderscore S\textunderscore  4642  &    296878           &       30NC \\

128\textunderscore S\textunderscore 0135        &     10431        &     15MCI &
033\textunderscore S\textunderscore  0516  &    15860           &       15NC &  
 012\textunderscore S\textunderscore   4643 &      297693         &       30NC \\

128\textunderscore S\textunderscore 0138        &     10438        &     15MCI &
002\textunderscore S\textunderscore  0559  &    15948           &       15NC &
  029\textunderscore S\textunderscore  4585  &      298523         &       30NC \\

 098\textunderscore S\textunderscore 0160        &     10466        &     15MCI &  
 014\textunderscore S\textunderscore  0548  &  16024             &       15NC &
 013\textunderscore S\textunderscore  4616  &      300089         &       30NC \\

123\textunderscore S\textunderscore 0108        &     10738        &     15MCI & 
128\textunderscore S\textunderscore  0545  &     16090          &       15NC &    
029\textunderscore S\textunderscore  4652  &     300886          &       30NC \\

037\textunderscore S\textunderscore 0150        &     10773       &     15MCI & 
010\textunderscore S\textunderscore 0420   &     17078          &       15NC & 
 137\textunderscore S\textunderscore 4632   &      301677         &       30NC \\
  
   027\textunderscore S\textunderscore 0116        &      10783       &     15MCI &
 126 \textunderscore S\textunderscore  0506  &    17184           &       15NC &  
   094\textunderscore S\textunderscore  4649  &     302926          &       30NC\\

128\textunderscore S\textunderscore  0188       &      10897       &     15MCI &    
005\textunderscore S\textunderscore  0610  &     17303          &       15NC &
 016\textunderscore S\textunderscore   4638 &       305882        &       30NC \\

014\textunderscore S\textunderscore 0169        &      10987      &      15MCI & 
 006\textunderscore S\textunderscore  0484  &    17377           &       15NC &  
 013\textunderscore S\textunderscore  4731  &     308178          &       30NC \\

021\textunderscore S\textunderscore 0178        &       10993     &      15MCI &  
031\textunderscore S\textunderscore  0618  &     16598          &       15NC &   
136\textunderscore S\textunderscore  4726  &    308396           &       30NC\\

128\textunderscore S\textunderscore 0205        &       11011      &      15MCI &   
 016\textunderscore S\textunderscore 4951   &    337692           &       30NC &
016\textunderscore S\textunderscore   4688 &    310327           &       30NC \\

128\textunderscore S\textunderscore 0200        &       11012      &      15MCI & 
003\textunderscore S\textunderscore  4839  &   319414            &       30NC &
 019\textunderscore S\textunderscore  4835  &  315857             &       30NC \\

037\textunderscore S\textunderscore 0182        &       11121      &      15MCI &  
  003\textunderscore S\textunderscore  4900  &   325729            &       30NC &
 127\textunderscore S\textunderscore 4843   &    316771           &       30NC  \\
 
   003\textunderscore S\textunderscore  4840  &     319427          &       30NC & 
  003\textunderscore S\textunderscore  4872  &   321376            &       30NC & 
& &  \\



\bottomrule
\end{longtable}
\end{center}

\end{document}